\documentclass[final,authoryear,times]{elsarticle}
\usepackage{url,hyperref}
\usepackage[english]{babel}
\usepackage{graphicx} 
\usepackage{subcaption}
\usepackage{amsmath,amsthm,amsfonts,bbm,bm}
\usepackage{hyperref}
\usepackage{todonotes}
\setuptodonotes{fancyline,color=green!50}
\graphicspath{{Images/}}
\usepackage{float}

\theoremstyle{plain}
\newtheorem{thm}{Theorem}

\newtheorem{prop}[thm]{Proposition} 
\theoremstyle{definition}

\theoremstyle{remark}
\newtheorem{rem}[thm]{Remark}

\newcommand{\RR}{\mathbbm{R}}
\newcommand{\NN}{\mathbbm{N}}

\begin{document}

\title{Mathematical Modeling of Early Embryonic Cell Cycles of \emph{Drosophila melanogaster}}

\author[math]{Meskerem Abebaw Mebratie}\ead{mebratie@mathematik.uni-kassel.de}
\author[bio]{Benedikt Drebes}\ead{Benedikt-Drebes@uni-kassel.de}
\author[bio]{Katja Kapp}\ead{k.kapp@uni-kassel.de}
\author[bio]{Arno M\"uller}\ead{h.a.muller@uni-kassel.de}
\author[math]{Werner M. Seiler}\ead{seiler@mathematik.uni-kassel.de}

\address[math]{Institut f\"ur Mathematik, Universit\"at Kassel, 34109 Kassel}
\address[bio]{Institut f\"ur Biologie, Universit\"at Kassel, 34109 Kassel}

\date{\today}

\begin{abstract}
  In the early stages of development, \textit{Drosophila melanogaster}
  embryos possess very fast and well-coordinated cell cycles. In the cell
  cycle, CDK activity is essentially regulated by binding CDK and CycB to
  form an active complex and by phosphorylating CDK via CDC25 and
  dephosphorylating it via Wee1.  We develop a mathematical model for the
  embryonic cell cycle which is biochemically sound and which can be
  rigorously analysed after a model reduction.  We show that there exists a
  region in the parameter space where the model describes oscillations.  We
  then focus on the role of two parameters: the CycB synthesis and the
  activation coefficient of APC.  Our main biological hypothesis is that
  the first one is responsible for the period lengthening over the first 14
  cycles which can be experimentally observed and this hypothesis is
  supported by numerical simulations of our model: if the CycB synthesis is
  made time-dependent with a prescribed dynamics, then our simulations show
  qualitatively a very similar behavior to experimental data reported in
  the literature.
\end{abstract}

\maketitle

\section{Introduction}\label{II}

Life is characterized by several properties, including growth, metabolism,
and the ability to reproduce.  The ability to reproduce is essential for
complex organisms as well as individual cells.  At a cellular level, this
attribute is enabled by the cell cycle. The cell cycle controls cell
growth, the duplication of genetic information, and cell division. The cell
cycle is divided into two phases: interphase and mitosis.  The interphase
is further subdivided into two gap phases (G phases) and the synthesis
phase (S phase) in between.  G1 is the gap phase during which cells grow
and prepare for DNA synthesis.  During the S phase, the centrosomes
duplicate and the DNA is replicated.  G2 is the second gap phase, during
which the cell continues to grow and prepares for cell division.  After the
G2 phase, the cell enters mitosis.  During mitosis, the replicated
chromosomes are distributed into two daughter cells.  This process can be
subdivided into five phases: prophase, metaphase, anaphase, telophase, and
cytokinesis \citep{matthews2022cell}.

The embryo of the fruit fly \emph{Drosophila melanogaster} has been used
extensively to study various cell biological aspects of the cell cycle.
This model offers several advantages, such as the ability to use a range of
genetic tools to investigate the function of individual proteins. The first
13 cell cycles in \emph{D. melanogaster} are characterized by the absence
of G1 and G2 phases, and by the absence of cytokinesis and thus provides a
simplified model for the regulation of M and S phase cycling
\citep{yuan2016timing}.  Since the cell cycle consists only of the S phase
and mitosis, this cycle can be used to study the regulation of these two
phases (see Figure \eqref{RPa}).
  
The choice of this family of regulatory networks was motivated by the
investigation of early embryonic cell cycles in the model organism
\emph{D. melanogaster}, which has been widely used in developmental biology
\citep{wilson1995drosophila}.  In the early stages of development,
\emph{D. melanogaster} embryos display very fast and well-coordinated cell
cycles.  Precise temporal controls and other intrinsic molecular mechanisms
govern these cycles.  Numerous developmental biology laboratories have
studied these mechanisms, and many biologists have contributed to this
effort \citep{deneke2016waves}, \citep{farrell2014egg},
\citep{shindo2021excess}, \citep{shindo2021modeling},
\citep{yuan2016timing}.

\begin{figure}[H]
  \centering
  \begin{subfigure}{0.5\textwidth}
    \centering
    \includegraphics[width=\textwidth]{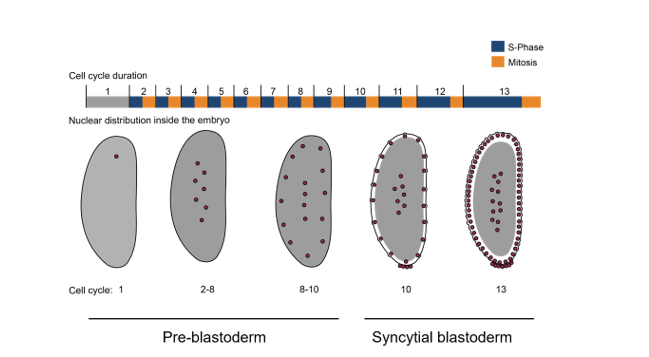}
    \caption{The timing of the cell cycles and the first 13 cell cycles of
      \emph{D. melanogaster} along with nuclei position, modified according
      to \citep{farrell2014egg}.}\label{RPa}
  \end{subfigure}\quad
  \begin{subfigure}{0.45\textwidth}
    \centering
    \includegraphics[width=\textwidth]{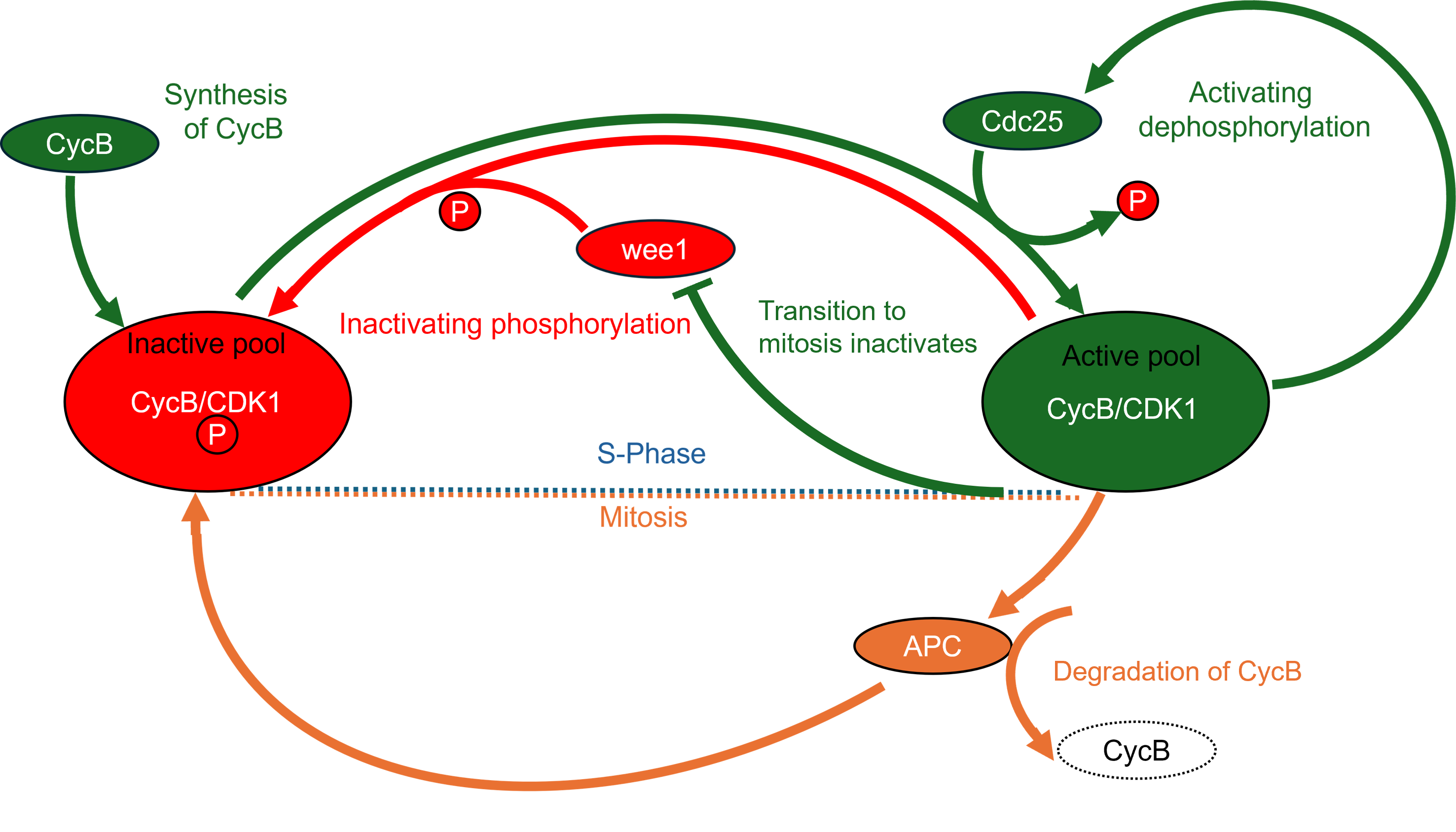}
    \caption{Regulatory processes impinging on $CDK1$ kinase activity.}\label{RPb}
  \end{subfigure}
   \caption{Timing and regulation of the early embryonic \emph{D. melanogaster}.}
  \label{RPDeme}
\end{figure}

In this work, we develop a biochemically sound model of CDK activation in
the form of a parametric system of ordinary differential equations and
analyze it mathematically.  One of the main questions in the analysis is
the existence of a region of positive parameter values for which the system
possesses limit cycles, i.e.\ shows oscillatory behavior.  For systems in
more than two dimensions, a direct proof of the existence of a limit cycle
is known to be a challenging problem, in particular for a larger number of
parameters.  We therefore instead identify parameter values at which a Hopf
bifurcation \citep{marsden2012hopf}, \citep{rionero2019hopf} occurs, as
limit cycles in parametric systems typically stem from such a bifurcation.
In experiments \citep{deneke2016waves}, one observes that the length of the
early embryonic cell cycles increases from cycle to cycle and our
biological hypothesis is that this effect is controlled by the parameter
representing the cyclin synthesis.  We show that our model can reproduce
such a behavior, if one assumes that the cyclin synthesis is slowly
decreasing and thus provide mathematical support for the biological
hypothesis.

The paper is organized as follows: In Section~\ref{III}, we develop our
mathematical model for the early embryonic cell cycles of \emph{D.
  melanogaster}.  In Section~\ref{IV}, we apply several reduction
techniques to obtain a smaller model that can be better analyzed
mathematically.  The analysis is the topic of Section~\ref{sec:bifana}: we
first search for biologically relevant steady states of the reduced system
and then perform a bifurcation analysis for two key parameters.  In
Section~\ref{exr}, we compare numerical simulations of our model with
experimental results. Finally, some conclusions are given.

\section{Model Development}\label{III}

For the cell cycle to proceed correctly, cells require precise temporal
control. All the molecules that control M and S phases are present in the
fly embryo as mRNA or proteins within a common cytoplasm.  Cell cycle
progression depends on the synthesis and destruction of Cyclins and the
transition from the S phase to M phase is controlled by the complex of
Cyclin B (CycB) and Cyclin-dependent kinase 1 (Cdk1)
\citep{pluta2024cyclin}.  The kinase activity of the CycB/Cdk1 complex is
both activated and inactivated via phosphorylation \citep{blank2025cdk}.
The inactivating phosphorylations are carried out by the Wee1 kinase (Wee1)
at the Threonine 14 (T14) and Tyrosine 15 (Y15)
\citep{stumpff2004drosophila}.  These inhibitory phosphorylations can be
removed by the phosphatase Cell Division Cycle 25 (Cdc25), thereby
activating the complex \citep{lew1996regulatory}. The activated CycB/Cdk1
complex in turn activates Cdc25 and thus triggers a positive feedback loop
\citep{lu2012multisite}.  When a certain amount of the CycB/Cdk1 complex is
activated, the cell can enter mitosis. The transition to mitosis
inactivates Wee1.  During mitosis, the CycB/Cdk1 complex activates the
anaphase-promoting complex (APC) via phosphorylation of Cdc20, which
eventually results in the activation of APC \citep{kataria2019interplay}.
The activation of APC allows the cell to proceed from metaphase to
anaphase.  The activated APC degrades CycB via ubiquitylation and thereby
destroys the CycB/Cdk1 complex \citep{acquaviva2006anaphase}.  At the onset
of the next cell cycle, CycB is translated again from maternally provided
CycB mRNA and will form a novel complex with Cdk1; therefore the
translation and destruction of CycB represent the temporal regulation of
this simple cell cycle (see Figure \eqref{RPb}).

The simplification of the syncytial cell cycles in the fly embryo to only S
and M phases provides the basis for very rapid cell cycles with a duration
of about 9 minutes.  From the 10th cell cycle onwards, the cell cycle
durations increase gradually to a length of eventually 21 minutes for the
13th cell cycle \citep{foe1983studies}.  Coincindently with the extension
of the cell cycle length, a decrease in the amount of maternal CycB mRNA
can be observed at the beginning of the cell cycle \citep{edgar1994early}.
The decrease in CycB concentration and the regulation of CycB/Cdk1 by Wee1,
Cdc25 and APC are used to create a model that explains the change in the
cell cycle length.  All of these signals exhibit significant time delays in
their functioning, which seem to have important biological implications.
Moreover, the CDK network is susceptible to external checkpoint controls
\citep{murray1992creative}.

The regulatory system, as depicted in Figure~\ref{RPDeme}, will be
transformed into a system of ordinary differential equations.  As
illustrated in Figure~\ref{RPDeme}, we assume that the substrates (CDC25,
Wee1) play a major role over the proteins (phosphatases and kinases) and
that other proteins have indirect effects on CDK activity.  If not, the
substrates competing for the same kinase would interfere in intricate ways.
These assumptions make the model much simpler.  We use Hill functions to
model activating and inhibiting effect in a standard way.  As a result, the
Hill functions we apply to describe the dynamics of CDC25, Wee1, and APC
depend on their roles in the CDK activity.

Hence, the equation for the protein CDC25 is
$\frac{d[CDC25]}{dt}=\theta_5\frac{[CDK]^{n_2}}{k_2^{n_2}+[CDK]^{n_2}}-\alpha_3[CDC25]$
with the activation coefficient $k_2$, the maximal expression level of the
promoter $\theta_5$, the degradation rate $\alpha_{3}$, and the Hill
coefficient $n_2$.  For the protein Wee1, that the CDK-cyclin active
complex inhibits, we use
$\frac{d[Wee1]}{dt}=\theta_6\frac{k_1^{n_1}}{k_1^{n_1}+[CDK]^{n_1}}-\alpha_4[Wee1]$
where the inhibition rate is $\theta_6$, the inhibition coefficient $k_1$,
the degradation rate $\alpha_4$, and the Hill coefficient $n_1$, which
measures the process switch-like nature.  The protein $APC$ is activated by
the active complex $CDK$ and its dynamic is modelled as
$\frac{d[APC]}{dt}=\theta_4\frac{[CDK]^{n_3}}{k_3^{n_3}+[CDK]^{n_3}}-\alpha_4[APC]$
where $k_3$ is the activation coefficient, $\theta_4$ the maximal level of
the promoter, $\alpha_2$ the degradation rate and $n_3$ the Hill
coefficient.  The difference between the rates of synthesis and degradation
of cyclin determines the rate at which its concentration changes over
time. When APC is present, cyclin is degraded; this relationship can be
explained by the following equation,
$\frac{d[\text{CycB}]}{dt} = b_1 - b_7[\text{CycB}][\text{APC}]$ where the
rates of cyclin synthesis and degradation in the presence of APC are $b_1$
and $b_7$, respectively.

Considering all the above, the mathematical model for the central
CDK--cyclin system is then described by the following differential system
\begin{align}
    \begin{split}
        \frac{d[CycB]}{dt}&=b_1-b_7[CycB][APC]\,,\\
        \frac{d[CDK]}{dt}&=b_2[CycB][CDK_I]+\theta_2[CDC25][CDK_I]\\
        &-\Bigl(\theta_3[Wee1]+\alpha_1[APC]\Bigr)[CDK]= -\frac{d[CDK_I]}{dt}\,,\\
       \frac{d[APC]}{dt}&=\theta_4\frac{[CDK]^{n_3}}{[CDK]^{n_3}+k_3^{n_3}}-\alpha_2[APC]\,,\\
        \frac{d[CDC25]}{dt}&=\theta_5\frac{[CDK]^{n_2}}{[CDK]^{n_2}+k_2^{n_2}}-\alpha_3[CDC25]\,,\\
        \frac{d[Wee1]}{dt}&=\theta_6\frac{k_1^{n_1}}{k_1^{n_1}+[CDK]^{n_1}}-\alpha_4[Wee1]\,,
    \end{split}
    \label{2mm}
\end{align}
where $CDK$ denotes the active $CDK$ and $CDK_I$ the inactive $CDK$. It is
a system of six differential equations depending on fifteen parameters and
three Hill coefficients.

\section{Model Reductions}\label{IV}

In the form \eqref{2mm}, the system is still too large for a rigorous
mathematical analysis.  We therefore apply various methods for a model
reduction to obtain a better analysable model.  We first note that $CDC25$
appears outside of its own equation only in the equations for
$CDK/CDK_{I}$.  As explained above, its effect on the $CDK$ activity is
described by the Hill function $\frac{[CDK]^{n_2}}{[CDK]^{n_2}+k_2^{n_2}}$
and hence we can replace the second term of the second equation of
\eqref{2mm} by $\theta_2\frac{[CDK]^{n_2}}{[CDK]^{n_2}+k_2^{n_2}}[CDK_I]$
effectively eliminating $CDC25$.  The equation for $CDK$ now becomes
\begin{multline*}
  \frac{d[CDK]}{dt}=
  b_2[CycB][CDK_I]+\theta_2\frac{[CDK]^{n_2}}{[CDK]^{n_2}+k_2^{n_2}}[CDK_I]-{}\\
  \Bigl(\theta_3[Wee1]+\alpha_1[APC]\Bigr)[CDK]
  =-\frac{d[CDK_I]}{dt}
\end{multline*}
and there is no need to keep a differential equation for $CDC25$.

Next we note that the equality $\frac{d[CDK]}{dt}=-\frac{d[CDK]}{dt}$
implies that the overall concentration $[CDK_{total}] =[CDK]+[CDK_I]$ is a
conserved quantity.  Writing $[CDK_I]$ as $[CDK_I]=[CDK_{total}] -[CDK] $
and dividing both sides by $[CDK_{total}] $, we can eliminate the
differential equation for $[CDK_I]$, leading to the system
\begin{align}
    \begin{split}
        \frac{d[CycB]}{dt}&=b_1-b_7[CycB][APC]\,,\\
        \frac{d[\overline{CDK}]}{dt}&=b_2[CycB]\Bigl(1-[\overline{CDK}]\Bigr)+
        \theta_2\frac{[\overline{CDK}]^{n_2}}{[\overline{CDK}]^{n_2}+k_2^{n_2}}\Bigl(1-[\overline{CDK}]\Bigr)\\
        &-\Bigl(\theta_3[Wee1]+\alpha_1[APC]\Bigr)[\overline{CDK}]\,,\\
       \frac{d[APC]}{dt}&=\theta_4\frac{[\overline{CDK}]^{n_3}}{[\overline{CDK}]^{n_3}+k_3^{n_3}}-\alpha_2[APC]\,,\\
        \frac{d[Wee1]}{dt}&=\theta_6\frac{k_1^{n_1}}{k_1^{n_1}+[\overline{CDK}]^{n_1}}-\alpha_4[Wee1]\,,
    \end{split}
    \label{2m}
\end{align}
where $[\overline{CDK}]=\frac{[CDK_I]}{[CDK_{total}]}$.  While it does not
explicitly depend on the new parameter $[CDK_{total}]$, reconstruction of
the values of $[CDK]$ and $[CDK_I]$ requires it.

In a further step, we apply scaling symmetries for reducing the number of
parameters (i.\,e., we rewrite the system using dimensionless quantities).
We used the \textsc{Maple} programme described in \citep{hubert2013scaling}
for an automated reduction.  As the programme cannot work with symbolic
exponents, we set all Hill coefficients $n_{k}$ equal to $2$ and afterwards
checked that the obtained scaling symmetries are correct also for other
values of the $n_{k}$.  In fact, this was rather obvious, as the reduction
does not affect the Hill functions because neither $[\overline{CDK}]$ nor
any parameter $k_{i}$ is rescaled.

For biological reasons, which we will address later, we also constrain our
reduction so that the parameters $b_1$ and $k_3$ remain unchanged in the
reduced system. Then, after the rescaling, there are only $10$ parameters
left, a reduction in the number of parameters by almost one-fourth.  The
new parameters are denoted by $c_i$, where $i=1,2,\ldots, 10$, and to
simplify the notation, we abbreviate from now the concentrations by
$x_1=[CycB]$, $x_2=[\overline{CDK}]$, $x_3=[APC]$, and $x_4=[Wee1]$ and
denote the rescaled concentrations by $y_{i}$.  The new parameters and
variables are defined as
$ c_1=b_1,\ c_2=\frac{b_2}{\theta_2^2},\ c_3=\frac{b_7\theta_4}{b_2},\
c_4=\frac{\theta_3\theta_6}{b_2},\ c_5=k_1,\ c_6= k_2,\ c_7=k_3,\
c_8=\frac{b_2\alpha_1}{\theta_2^2b_7},\ c_9=\frac{\alpha_2}{\theta_2},\
c_{10}=\frac{\alpha_4}{\theta_2}, \ \tau=\theta_2t,\ \ y_1=\theta_2x_1,\ \
y_2=x_2,\ y_3=\frac{b_7\theta_2}{b_2}x_3,\ \
y_4=\frac{\theta_2\theta_3}{b_2}x_4\,$.  We arrive now at the rescaled
system:
\begin{align}
\begin{split}
\frac{dy_1}{d\tau}&=c_1-c_2y_1y_3\,,\\
\frac{dy_2}{d\tau}&=c_2y_1+\frac{y_2^{n_2}}{y_2^{n_2}+c_6^{n_2}}\\
&-\left(c_2y_1+\frac{y_2^{n_2}}{y_2^{n_2}+c_6^{n_2}}+c_2y_4+c_8y_3\right)y_2\,,
\end{split}
&
\begin{split}
\frac{dy_3}{d\tau}&=c_3\frac{y_2^{n_3}}{y_2^{n_3}+c_7^{n_3}}-c_9y_3\,,\\
\frac{dy_4}{d\tau}&=c_4\frac{c_5^{n_1}}{c_{5}^{n_1}+y_2^{n_1}}-c_{10}y_4\,.
\end{split}
\label{rdu}
\end{align}

All reductions performed so far are of an exact nature.  In a last step, we
exploit that the evolution of the cell cycle of \emph{D. melanogaster}
embryogenesis involves both slow and fast time scales which allows for the
use of approximate reduction techniques like the quasi-steady state
approximation, see e.\,g.\ \citep{bertram2017multi} or
\citep{kuehn2015multiple}.  According to
\citep[Sect.~1.3.1]{alon2019introduction}, there is a time response
separation among transcription networks: input signals typically influence
transcription factor activities on a subsecond period of time, and it
frequently takes a few seconds for the active transcription factor's
binding to its DNA sites to reach equilibrium.  The translation and
transcription of the target gene take minutes, and the accumulation of the
protein product can take several minutes to many hours. This leads to a
large variation in the time response between the signal and the
accumulation of protein products.  On the slow timescale, transcription
networks can be thought of as being in a steady state, whereas signaling
transduction networks composed of interacting proteins function at a far
slower pace.  In our case the cyclin time response is faster than the other
ones, so that we may assume that $[CycB]$ is in a quasi-steady state.

To make the time scale separation better visible, we perform a further
rescaling of the variables $y_1, y_3$, and $y_4$ to
\begin{align*}
  z_1=\frac{c_2}{c_1}y_1,\quad
  z_3=\frac{c_9}{c_3}y_3,\quad
  z_4=\frac{c_{10}}{c_4}y_4\,.
\end{align*}
This transforms \eqref{rdu} to the system
\begin{align}
\begin{split}
\frac{dz_1}{d\tau}=&c_2(1-\beta_1z_1z_3)\,,\\
\frac{dz_2}{d\tau}=&c_1z_1+\frac{z_2^{n_2}}{z_2^{n_2}+c_6^{n_2}}\\
&-\left(c_1z_1+\frac{z_2^{n_2}}{z_2^{n_2}+c_6^{n_2}}+\beta_2z_4+\beta_3z_3\right)z_2\,,\\
\end{split}
&
\begin{split}
\frac{dz_3}{d\tau}=&c_9\left(\frac{z_2^{n_3}}{z_2^{n_3}+c_7^{n_3}}-z_3\right)\,,\\
\frac{dz_4}{d\tau}=&c_{10}\left(\frac{c_5^{n_1}}{c_{5}^{n_1}+z_2^{n_1}}-z_4\right)\,,
\end{split}
\label{420}
\end{align}
where
$\beta_1=\frac{c_3}{c_9}, \ \ \beta_2=\frac{c_2c_4}{c_{10}}, \ \
\beta_3=\frac{c_3c_8}{c_{9}}$.  In this formulation,
$\epsilon=\frac{1}{c_2}=\frac{\theta_2^2}{b_2}$ is the parameter leading to
the time scale separation, i.\,e.\ it should be small enough so that the
variable corresponding to cyclin ($z_1$) is in a quasi-steady state.  Now
the differential equation for cyclin is given by
$\epsilon\frac{dz_1}{d\tau}=(1-\beta_1z_1z_3)$ and in the limit
$\epsilon\to0$, corresponding to the quasi-steady state approximation, we
get $z_1=1/(\beta_1z_3)\,.$ This result agrees with the biological
assumption that $z_1$ is degraded by $z_3$.  Entering it into system
\eqref{420}, we obtain our final reduced model
\begin{align}
    \begin{split}
      \frac{dz_2}{d\tau}=&\frac{c_1}{\beta_1z_3}+\frac{z_2^{n_2}}{z_2^{n_2}+c_6^{n_2}}\\
      &-\left(\frac{c_1}{\beta_1z_3}+\frac{z_2^{n_2}}{z_2^{n_2}+c_6^{n_2}}+\beta_2z_4+\beta_3z_3\right)z_2\,,
    \end{split}
    &
    \begin{split}
\frac{dz_3}{d\tau}=&c_9\left(\frac{z_2^{n_3}}{z_2^{n_3}+c_7^{n_3}}-z_3\right)\,,\\
\frac{dz_4}{d\tau}=&c_{10}\left(\frac{c_5^{n_1}}{c_{5}^{n_1}+z_2^{n_1}}-z_4\right)\,,
    \end{split}
    \label{3}
\end{align}
a system of three differential equations depending on nine parameters.  In
fact, the two parameters $c_1$ and $\beta_1$ appear only as the fraction
$\frac{c_1}{\beta_1}$.  Hence, the final model could be described with only
$8$ parameters, if this fraction was treated as a single parameter.
However, we will later perform a bifurcation analysis with $c_1$ as
bifurcation parameter, as $c_{1}$ is the key parameter for our biological
hypothesis about the role of the cyclin synthesis, and therefore we leave
the combination $\frac{c_1}{\beta_1}$ as it is.

\section{Bifurcation Analysis of the Reduced Model}
\label{sec:bifana}

As our reduced model~\eqref{3} still depends on nine parameters, a complete
bifurcation analysis is not possible.  Instead, in view of our biological
question, we will concentrate on two parameters, namely on $c_{1}=b_{1}$
and $c_{7}=k_{3}$, i.\,e.\ on the cyclin synthesis and on the activation
coefficient for APC.  Below, we will perform for each of these two
parameters an individual bifurcation analysis.

\subsection{Biologically Relevant Steady States and Their
  Stability}\label{pss}

We begin the mathematical analysis of the final reduced system~\eqref{3} by
searching for biologically relevant steady states.

\begin{prop}
  For any parameter vector $\theta\in\RR^{9}$ and for any integer values of
  the Hill coefficients $n_{1},n_{2},n_{3}\in\NN$, the reduced system
  \eqref{3} possesses at least one positive steady state
  $\left(\hat{z}_2, \hat{z}_3, \hat{z}_4\right)\in\RR_+^3$.
\end{prop}

\begin{proof}
  The steady states are given by the common zeros of the right-hand sides
  of~ \eqref{3} which we denote by $(f_{1},f_{2},f_{3})$.  Since $f_{2}$
  and $f_{3}$ are linear in $\hat{z}_3$ and $\hat{z}_4$, respectively,
  we can express these two coordinates by $\hat{z}_2$:
  \begin{equation}
    \hat{z}_3=\frac{\hat{z}_2^{n_3}}{\hat{z}_2^{n_3}+c_7^{n_3}}\,,\qquad
    \hat{z}_4=\frac{c_5^{n_1}}{c_{5}^{n_1}+\hat{z}_2^{n_1}}\,.
    \label{e416}
  \end{equation}
  Obviously, these expressions yield positive values whenever $\hat{z}_2$
  is positive.  Entering them into the first equation yields a rational
  equation for $\hat{z}_2$ alone which after multiplication by the common
  denominator leads to the following polynomial equation for $\hat{z}_2$
  where the degrees of the different terms depend on the Hill coefficients:
  \begin{multline}
\left(- \beta_1\beta_3- \beta_1- c_1\right)\hat{z}_2^{n_1 + n_2 + 2n_3 + 1}+\left(\beta_1
+ c_1\right)\hat{z}_2^{n_1 + n_2 + 2n_3}\\
{}+\left(- \beta_1\beta_2c_5^{n_1}
- \beta_1\beta_3c_5^{n_1}
- \beta_1c_5^{n_1} 
- c_1c_5^{n_1}\right)\hat{z}_2^{n_2 + 2n_3 + 1}\\
{}+\left(- \beta_1\beta_3c_6^{n_2}
- c_1c_6^{n_2}\right)\hat{z}_2^{n_1 + 2n_3 + 1}+\left(- \beta_1c_7^{n_3} 
- 2c_1c_7^{n_3}\right)\hat{z}_2^{n_1 + n_2 + n_3 + 1}\\
{}+\left(\beta_1c_7^{n_3}
+ 2c_1c_7^{n_3}\right)\hat{z}_2^{n_1 + n_2 + n_3}+\left(\beta_1c_5^{n_1}
+ c_1c_5^{n_1}\right)\hat{z}_2^{n_2 + 2n_3}+ c_1c_6^{n_2}\hat{z}_2^{n_1 + 2n_3}\\
{}+\left(- \beta_1\beta_2c_5^{n_1}c_7^{n_3}- \beta_1c_5^{n_1}c_7^{n_3}- 2c_1c_5^{n_1}c_7^{n_3}\right)\hat{z}_2^{n_2 + n_3 + 1}- 2c_1c_6^{n_2}c_7^{n_3}\hat{z}_2^{n_1 + n_3 + 1}\\
{}+\left(- \beta_1\beta_2c_5^{n_1}c_6^{n_2}- \beta_1\beta_3c_5^{n_1}c_6^{n_2}- c_1c_5^{n_1}c_6^{n_2}\right)\hat{z}_2^{2n_3 + 1}- c_1c_7^{2n_3}\hat{z}_2^{n_1 + n_2 + 1}+ 2c_1c_6^{n_2}c_7^{n_3}\hat{z}_2^{n_1 + n_3}\\
{}+ c_1c_7^{2n_3}\hat{z}_2^{n_1 + n_2}+\left( \beta_1c_5^{n_1}c_7^{n_3} + 2c_1c_5^{n_1}c_7^{n_3}\right)\hat{z}_2^{n_2 + n_3}+ c_1c_5^{n_1}c_6^{n_2}\hat{z}_2^{2n_3}- c_1c_6^{n_2}c_7^{2n_3}\hat{z}_2^{n_1 + 1}\\
{}- c_1c_5^{n_1}c_7^{2n_3}\hat{z}_2^{n_2 + 1}+\left(-\beta_1\beta_2c_5^{n_1}c_6^{n_2}c_7^{n_3}
- 2c_1c_5^{n_1}c_6^{n_2}c_7^{n_3}\right)\hat{z}_2^{n_3 + 1} + 2c_1c_5^{n_1}c_6^{n_2}c_7^{n_3}\hat{z}_2^{n_3}+ c_1c_5^{n_1}c_7^{2n_3}\hat{z}_2^{n_2}\\
{}+ c_1c_6^{n_2}c_7^{2n_3}\hat{z}_2^{n_1}- c_1c_5^{n_1}c_6^{n_2}c_7^{2n_3}\hat{z}_2 + c_1c_5^{n_1}c_6^{n_2}c_7^{2n_3}=0\,.
    \label{e1}
  \end{multline}

  Our claim is true, if and only if this polynomial has at least one real
  positive root.  Independent of the chosen Hill coefficients, the term of
  highest degree is $\hat{z}_2^{n_1+n_2+2n_3+1}$ and it appears with a
  negative coefficient so that the polynomial becomes negative for
  sufficiently large $\hat{z}_{2}$.  As the constant term
  $c_1c_5^{n_1}c_6^{n_2}c_7^{2n_3}$ -- and thus the value
  of the polynomial for $\hat{z}_{2}=0$ -- is positive, there must be at
  least one real positive root.
\end{proof}
  
A more delicate question is whether this steady state is unique.
Furthermore, since the variable $z_{2}$ corresponds to a percentage, its
biologically meaningful values are confined to the interval $[0,1]$ and
thus we have to prove that the polynomial in \eqref{e1} has a root in this
interval.  We will now prove the existence of a parameter region
$\Theta\subseteq\RR^{9}$ where this is the case.  We cannot rigorously
describe this region.  Instead, we will show that for a certain parameter
vector $\hat{\theta}$ and for certain Hill coefficients, the polynomial in
\eqref{e1} possess only one real root in the interval $(0,1)$.  Considering
also the Hill coefficients as continuous quantities, this implies the
existence of an open domain in $\RR^{12}$ where the positive steady state
is unique and biologically meaningful.

The literature does not contain results on the direct measurement of the
parameters in our model, but one can find some numerical estimates stemming
from experimental data.  We considered values published in
\citep{deneke2016waves}, \citep{shindo2021modeling} and
\citep{pomerening2005systems} where the effects of different proteins --
including CDC25, Wee1, and APC -- on the CDK dynamics were studied.  We did
not use the exact values from these references, but used them as starting
points for adapting them to our biological context.  The finally used
values were also chosen with the numerical simulations in mind which we
will discuss later so that our values lead to results similar to the
experimentally observed cell cycle regulation for early embryogenesis of
\emph{D. melanogaster}.  In the cited literature, all Hill
coefficients $n_{k}$ were set to~$5$.  For this value the polynomial in
\eqref{e1} has degree $21$ and consequently all subsequent analysis becomes
very involved.  We therefore used instead $2$ as common value for all Hill
coefficients $n_{k}$ leading to a degree $9$ polynomial.  Our complete
parameter set can be found in Table~\ref{table 1}.

\begin{table}[t]
  \centering
  \caption{Parameter values used for analysis and simulations.}
  \label{table 1}
  \begin{tabular}{|c|c|c|}\hline
    $c_1=0.002$&$c_5=0.02$&$c_6=0.25$ \\
    $c_7=0.25$ & $c_9=0.42857$&$c_{10}=0.135857$\\
    $\beta_1=1.75$&$\beta_2=0.150217$ & $\beta_3=7.142857$\\
    $n_1=2$&$n_2=2$&$n_3=2$\\\hline
  \end{tabular}
\end{table}

We claim now that for the parameter values given in Table~\ref{table 1},
there exists only one real steady state and it is biologically relevant.
Entering the chosen parameter values into \eqref{e1}, one obtains the
following polynomial equation of degree nine (for the printing the
coefficients have been truncated to six digits):
\begin{multline}\label{ne1}
  -14.252\hat{z}_2^9 + 1.752\hat{z}_2^8 - 0.896805\hat{z}_2^7\\
  {}+ 0.110450\hat{z}_2^6 - 0.000392\hat{z}_2^5 +
  6.73375\cdot10^{-5}\hat{z}_2^4 - 9.08408\cdot10^{-7}\hat{z}_2^3\\
  {}+ 4.97656\cdot10^{-7}\hat{z}_2^2 - 1.95312\cdot10^{-10}\hat{z}_2 +
  1.95312\cdot10^{-10}\,.
\end{multline}

With standard numerical techniques (we used the \texttt{roots} method from
\textsc{NumPy} which estimates the roots as the eigenvalues of the
companion matrix), one can compute its nine roots which are all isolated
and simple.  Only one root is real and it also lies in the interval
$[0,1]$.  It leads to the unique steady state
\begin{equation}\label{uss}
  \hat{z}_2=0.12555\,,\quad
  \hat{z}_{3}=0.20140\,,\quad 
  \hat{z}_{4}=0.024748\,.
\end{equation}
Thus we can conclude that for the chosen parameter values, our model has
only one real steady state which is also biologically relevant.

Given that the last two coefficients of the polynomial \eqref{ne1} are very
small (indicating the presence of two roots close to zero), one may wonder
whether numerical errors may have lead here to wrong results.  Since
\textsc{Python} always computes in double precision, the two coefficients
are well distinguishable from zero.  We checked the residuals of all nine
roots and they are all smaller than $10^{-20}$.  In addition, we also
computed a Sturm sequence \citep[Chapt.~2]{jmm:roots} for the polynomial
\eqref{ne1} which also confirmed that there is exactly one real root in the
interval $(0,1]$.  The computation of such a sequence is sometimes
ill-conditioned, but as two different approaches lead to the same
conclusion about the existence of exactly one biologically relevant steady
state, we are very confident that it is correct.

We computed the roots of the polynomial \eqref{e1} also for three other
sets of parameter values.  Compared to Table~\ref{table 1}, we always only
changed either the value of $c_{1}$ or the value of $c_{7}$ (the choice of
the values will become apparent below).  In each case, the polynomial
\eqref{e1} has only one real root and it always lies in the interval
$(0,1)$.  Table~\ref{table 2} contains the corresponding locations of the
steady state.

\begin{table}[t]
  \centering
  \caption{Steady states for different parameter values}
  \label{table 2}
  \begin{tabular}{|cc|ccc|}
    \hline
    $c_{1}$ & $c_{7}$ & $\hat{z}_{2}$\strut & $\hat{z}_{3}$& $\hat{z}_{4}$ \\
    \hline
    $0.002$ & $0.250$ & $0.12555$ & $0.20140$ & $0.024748$ \\
    $0.040$ & $0.250$ & $0.15401$ & $0.27511$ & $0.01658$\\
    $0.002$ & $0.040$ & $0.014352$ & $0.11405$ & $0.6601$ \\
    $0.002$ & $0.500$ & $0.25766$ & $0.20983$ & $0.005989$ \\
    \hline
  \end{tabular}
\end{table}

\begin{rem}
  Using advanced algorithms from real algebraic geometry
  \citep{bpr:algorag} like cylindrical algebraic decompositions, one can in
  principle determine for fixed Hill coefficients a semi-algebraic
  description of the parameter region $\Theta\subseteq\RR^{9}$ in which the
  polynomial equation~\eqref{e1} possesses exactly one real root in the
  interval $(0,1)$.  However, given the rather large number of parameters
  and the potentially high degree of the polynomial, this appears
  unfeasible in practice.
\end{rem}

We now look at the stability of a steady state
$\left(\hat{z}_2, \hat{z}_3, \hat{z}_4\right)$ using the Jacobian of the
reduced system \eqref{3}.  In a tedious, but straightforward computation,
one obtains the characteristic polynomial
$\lambda^3+\gamma_2\lambda^2+\gamma_1\lambda+\gamma_0=0$ where the
coefficients are given by
\begin{displaymath}
  \gamma_2=c_9+c_{10}-\delta_1\,,\quad
  \gamma_1=(c_9-\delta_1)c_{10}-\delta_2\delta_3-\delta_1c_9\,,\quad
  \gamma_0=-(\delta_1c_9-\delta_2\delta_3)c_{10}\,,
\end{displaymath}
with the shorthands
\begin{gather*}
  \delta_1= -\frac{c_1}{\beta_1\hat{z}_3} +
  \frac{k_2^{n_2}n_2\hat{z}_2^{n_2-1}(1-\hat{z}_2)}{\left(\hat{z}_2^{n_2}+k_2^{n_2}\right)^2}
  - \frac{\hat{z}_2^{n_2}}{\hat{z}_2^{n_2}+k_2^{n_2}} - \beta_2\hat{z}_4 - \beta_3\hat{z}_3\\
  \delta_2 =
  -\frac{c_1}{\beta_1}\frac{1-\hat{z}_2}{\hat{z}_3^{2}}-\beta_3\hat{z}_2\,,\qquad
  \delta_3=\frac{c_9n_3k_3^{n_3}\hat{z}_2^{n_3-1}}{\left(\hat{z}_2^{n_3}+k_3^{n_3}\right)^2}\,.
\end{gather*}

For the same parameter values as in Table~\ref{table 2}, we obtained at the
corresponding unique steady states the eigenvalues shown in
Table~\ref{table 3}.  We always found one real and two conjugate complex
eigenvalues.  While the real eigenvalue is always negative, the real part
of the complex eigenvalues is positive for our basic parameter set given in
Table~\ref{table 1} and negative for all other studied parameter sets.
Hence, in the first case the steady state is unstable and one can expect an
expanding oscillation nearby, whereas in the second case the steady state
is stable with a damped oscillation nearby.  This observation indicates the
appearence of bifurcations, if we change either $c_{1}$ or $c_{7}$ around
our basic parameter set.  We will now show that indeed in both cases Hopf
bifurcations exist leading to stable limit cycles.

\begin{table}[t]
  \centering
  \caption{Eigenvalues of the Jacobian at the steady state for different parameter values}
  \label{table 3}
  \begin{tabular}{|cc|cc|}
    \hline
    $c_{1}$ & $c_{7}$ & $\lambda_{1}$ & $\lambda_{2/3}$ \\
    \hline
    $0.002$ & $0.250$ & $-0.135439$  & $0.081083 \pm 0.866471i$ \\
    $0.040$ & $0.250$ & $-0.135671$ & $-0.282007 \pm 1.217285i$ \\
    $0.002$ & $0.040$ & $-0.134$  &  $ -0.45392\pm1.063i$ \\
    $0.002$ & $0.500$ & $-0.135$  & $-0.5054\pm1.0097i$  \\
    \hline
  \end{tabular}
\end{table}

\subsection{Bifurcation Analysis for the Parameter $c_1$}

There exists a simple geometric explanation of the difference in the
behavior for the two values of $c_{1}$ considered above.  Our reduced
model \eqref{3} is three-dimensional so that nullclines define surfaces in
$\RR^{3}$ making their geometry more difficult to analyze.  We therefore
follow an idea from \citep{msv:pnc} based on \citep{afs:mbh} using
``pseudo-nullclines'' which are plane curves: we decompose \eqref{3} into
two interconnected modules -- one consisting of $z_{2}$, $z_{4}$, the other
one of $z_{3}$:
\begin{equation}
  \begin{aligned}
    \frac{dz_2}{d\tau}&=f_{2}(z_{2},z_{4};z_{3})\,,\\
    \frac{dz_4}{d\tau}&=f_{4}(z_{2},z_{4})\,,
  \end{aligned}
  \qquad\qquad \frac{dz_3}{d\tau}=f_{3}(z_{3};z_{2})\,.
\end{equation}
Here, one should consider $z_{3}$ as a constant parameter in the left
column and $z_{2}$ as a constant parameter in the right column.  In the
decomposition, we could have swapped the role of $z_{3}$ and $z_{4}$.
However, in the above form we will obtain a $z_{2}$ pseudo-nullcline
relating $z_{2}$ and $z_{3}$, i.\,e.\ the biologically most relevant
quantities: CDK $(z_2)$ and APC $(z_3)$ which is equivalent to the input
signal CycB ($z_1=1/(\beta_1z_3)$).

Searching for a steady state of the left module leads to equations
\begin{displaymath}
  z_{4}=h(z_{2})\,,\qquad f_{2}\bigl(z_{2},h(z_{2});z_{3}\bigr)=0\,,
\end{displaymath}
where $h$ is given by the right equation in \eqref{e416}.  In the
$z_{2}$-$z_{3}$-plane, the $z_{2}$-pseudo-nullcline is now the curve
implicitly described by the second equation, i.\,e.\ by
\begin{equation}\label{nc2}
  \frac{c_1}{\beta_1z_3} + \frac{z_2^{n_2}}{z_2^{n_2}+c_6^{n_2}} -
  \left(\frac{c_1}{\beta_1z_3} + \frac{z_2^{n_2}}{z_2^{n_2}+c_6^{n_2}} + 
    \beta_2\frac{c_5^{n_1}}{c_{5}^{n_1} + z_2^{n_1}} + \beta_3z_3\right)z_2
  = 0\,. 
\end{equation}
Clearing the common denominator shows that it is a plane algebraic curve of
degree seven.  Because of the simple structure of our system, the
$z_{3}$-pseudo-nullcline is defined by the vanishing of $f_{3}$, i.\,e.\ by
the left equation in \eqref{e416}.

Above the $z_{2}$-pseudo-nullcline, CDK becomes inactive ($\dot{z}_{2}<0$)
and below it CDK is activated ($\dot{z}_{2}>0$).  Analogously, to the left
of the $z_{3}$-pseudo-nullcline, APC is degraded ($\dot{z}_{3}<0$) and to
the right of it synthesized ($\dot{z}_{3}>0$).  The
$z_{3}$-pseudo-nullcline is a monotonically increasing curve approaching a
saturation value for large $z_{3}$.  It is affected only by changes in the
parameter $c_{7}$ and even then only changes its steepness, but not its
basic shape.  By contrast, the $z_{2}$-pseudo-nullcline depends on six
parameters -- including $c_{1}$ -- and its shape is strongly affected by
changes in the parameter values.

\begin{figure}[t]
  \centering
  \begin{subfigure}{0.45\textwidth}
    \centering
    \includegraphics[width=\textwidth]{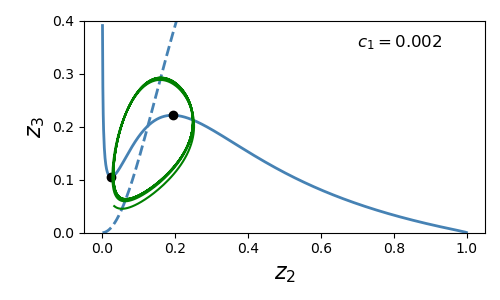}
  \end{subfigure}\qquad
  \begin{subfigure}{0.45\textwidth}
    \centering
    \includegraphics[width=\textwidth]{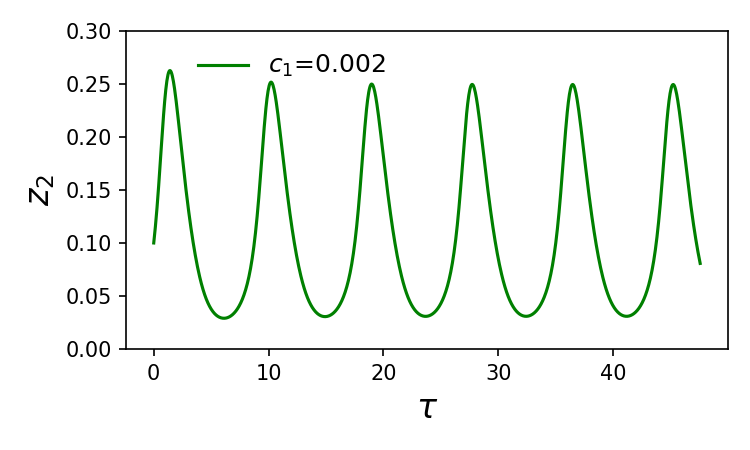}
  \end{subfigure}\\
  \begin{subfigure}{0.45\textwidth}
    \centering
    \includegraphics[width=\textwidth]{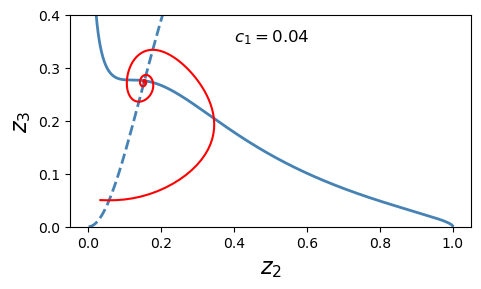}
  \end{subfigure}\qquad
  \begin{subfigure}{0.45\textwidth}
    \centering
    \includegraphics[width=\textwidth]{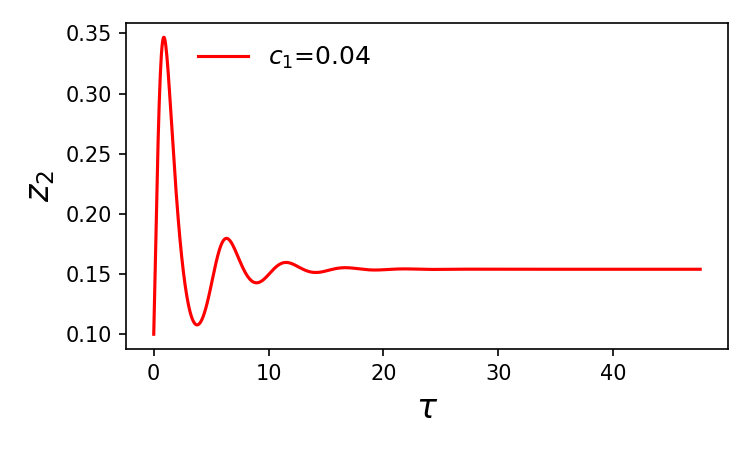}
  \end{subfigure}
  \caption{Phase portraits projected to $z_{2}$-$z_{3}$ plane and CDK
    evolution for two different values of $c_{1}$: top row $c_{1}=0.002$,
    bottom row $c_{1}=0.04$.  The left column shows a typical trajectory
    together with the pseudo-nullclines (solid -- $z_{2}$, dashed --
    $z_{3}$).  In the right column a time plot of $z_{2}$ is depicted
    showing that in the first case a sustained and in the second case a
    strongly damped oscillation occurs.}
  \label{pb1}
\end{figure}

In the left column of Figure~\ref{pb1}, one can see the pseudo-nullclines
for the two considered values of $c_{1}$ with their intersection defining
the steady state.  While the dashed $z_{3}$-pseudo-nullcline is the same in
both plots, the shape of the $z_{2}$-pseudo-nullcline changes drastically.
For the smaller value, it has an S-shape with a local minimum at
$(0.025, 0.10518)$ and a local maximum at $(0.193, 0.2214)$, whereas for
the larger value it is a monotonically decreasing curve.

Such an S-shaped curve is reminiscent of relaxation oscillators like the
FitzHugh--Nagumo model and indeed the phase portrait indicates the
existence of a limit cycle.  However, there are some differences probably
due to the fact that here the pseudo-nullcline is described by a septic and
not by cubic polynomial.  In the considered parameter domain, we always
have only one steady state and its stability is not simply determined from
whether it lies on an increasing or decreasing branch of the
$z_{3}$-pseudo-nullcline.  One can identify three critical values of the
cyclin synthesis $c_{1}$: above $c_{1}^{(s)}$ the S-shape of the
$z_{2}$-pseudo-nullcline disappears and it becomes a monotonically
decreasing curve; above $c_{1}^{(m)}$ the steady state lies to the right of
the maximum of the $z_{2}$-pseudo-nullcline and at $c_{1}^{(h)}$ a Hopf
bifurcation occurs and the stability of the steady state changes from
unstable to stable for higher values.

Numerically, all three critical values are readily determined.  The septic
polynomial describing the $z_{2}$-pseudo-nullcline is quadratic in $z_{3}$
and hence can easily be solved in the form $z_{3}=g(z_{2};c_{1})$.  The two
roots have different signs and biologically only the positive sign is
relevant; hence $g$ denotes from now on the positive solution.  The
critical value $c_{1}^{(s)}$ can be determined by solving the bivariate
system
\begin{equation}
  \frac{\partial g}{\partial z_{2}}(z_{2};c_{1})=
  \frac{\partial^{2} g}{\partial z_{2}^{2}}(z_{2};c_{1})=0
\end{equation}
leading to the value $c_{1}^{(s)}=0.039337$ (the $z_{2}$-value is not
interesting).  For determining the critical value $c_{1}^{(m)}$, we must
solve the bivariate system
\begin{equation}
  \frac{\partial g}{\partial z_{2}}(z_{2};c_{1})=
  g(z_{2};c_{1})-h(z_{2})=0
\end{equation}
where again $h$ is given by the right equation in \eqref{e416}.  One then
finds $c_{1}^{(m)}=0.031080$.  The Hopf point is as usually obtained by
monitoring the complex eigenvalues of the Jacobian of \eqref{3} at the
steady state and is given by $c_{1}^{(h)}=0.0079967$.  One also verifies
easily that one deals with a supercritical Hopf bifurcation.

\begin{figure}[t]
  \centering
  \includegraphics[width=0.5\textwidth]{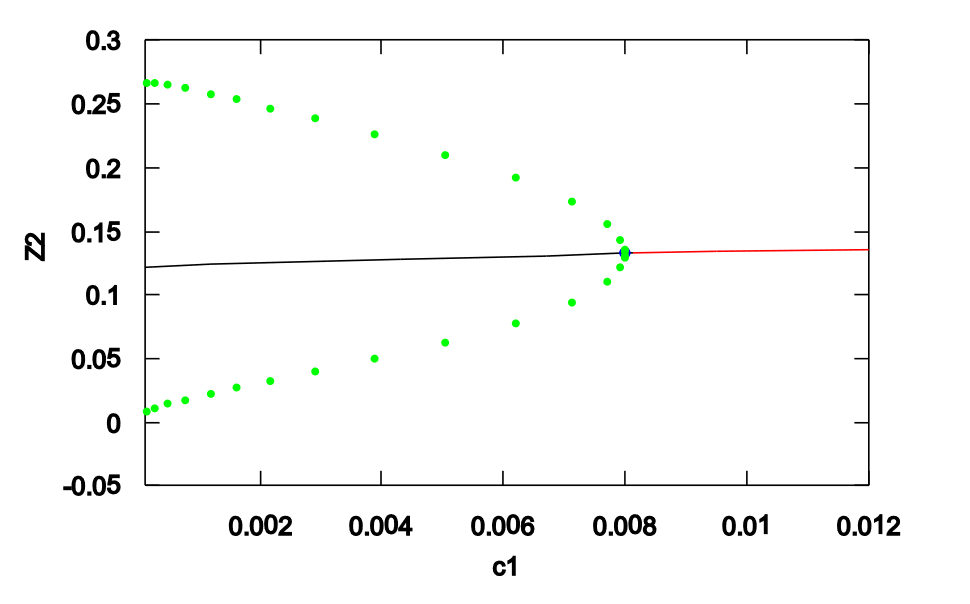}
  \caption{Bifurcation diagram for $c_1$.  Green dots indicate the size of
    the limit cycle; the black line shows unstable steady states, the red
    line stable ones.}
  \label{fig:bifc1}
\end{figure}

A bifurcation diagram in shown in Figure~\ref{fig:bifc1}.  For lower
$c_{1}$-values, a stable limit cycle (its size is indicated by green dots)
around an unstable steady state (shown in black) exists.  For higher
$c_{1}$-values, we have a stable steady state (shown in red).  Note that
the situation changes completely for the border case $c_{1}=0$, as now
$z_{2}=0$ becomes a multiple real root of the polynomial \eqref{e1}.  As
this case is biologically irrelevant, we ignore it.  The appearance of a
Hopf bifurcation is not very surprising, as the eigenvalues shown above for
two different $c_{1}$-values clearly indicate that between them a conjugate
pair of complex eigenvalues crosses the imaginary axis.  The left part of
Figure~\ref{cpee} shows for the lower value of $c_{1}$ a trajectory
approaching the limit cycle.  One can see that the approach concerns mainly
the $z_{4}$-coordinate, i.\,e.\ the Wee1 concentration.  Indeed, in the
projection to the $z_{2}$-$z_{3}$-plane shown in the upper left part of
Figure~\ref{pb1}, it looks like the trajectory would almost immediately
reach the limit cycle, but in the 3D plot one can see that this is not
really the case.  Of course, this is mainly due to our choice of the
initial value for the trajectory.

In the oscillatory regime, the phase portrait shown in the upper left part
of Figure~\ref{pb1} agrees well with the biological understanding of the
cell cycle.  Recall that by our quasi-steady state approximation the cyclin
concentration is the reciprocal of the APC concentration
($z_{1}=1/\beta_{1}z_{3}$).  If the system is in a state above both
pseudo-nullclines, then changes in the APC (or cyclin) concentration have
only a relative small effect on the CDK concentration.  If the APC
concentration crosses the $z_{2}$-pseudo-nullcline (which happens at a low
value and thus at a high cyclin level), then CDK activation starts and its
concentration increases rapidly.  But then the APC concentration is also
increasing and if it crosses again the $z_{2}$-pseudo-nullcline, then CDK
becomes inactivated and after some time the cycle starts again.  The
S-phase roughly corresponds to the part of the limit cycle around the left
crossing of the pseudo-nullcline (i.\,e.\ the troughs in the time
evolution) and the M-phase to the part around the right crossing (i.\,e.\
the crests in the time evolution).
  
The right part of Figure~\ref{cpee} shows the period length of the limit
cycle as a function of the cyclin synthesis $c_{1}$ in a log-log plot.  One
can see that it can be very well described by a power law of the form
$T=bc_{1}^{\alpha}$ with $b=0.43598$ and $\alpha=-0.186004$, i.\,e.\
roughly with the inverse of a fifth root.  Thus for low cyclin synthesis
the period length strongly growth.  This agrees well with our biological
hypothesis that a decrease in the cyclin synthesis is responsible for the
increase in the cycle length in embryogensis.

\begin{figure}[t]
  \centering
  \captionsetup{skip=65pt}
  \begin{subfigure}[c][0pt][c]{0.47\textwidth}
    \includegraphics[width=\linewidth]{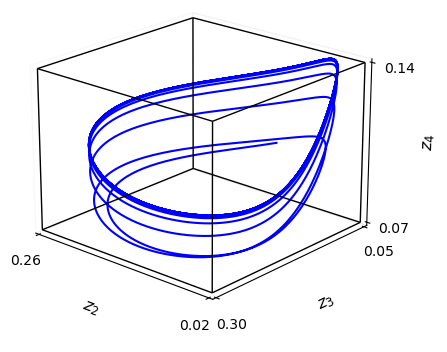}
  \end{subfigure}\quad
  \begin{subfigure}[c][0pt][c]{0.47\textwidth}
    \includegraphics[width=\textwidth]{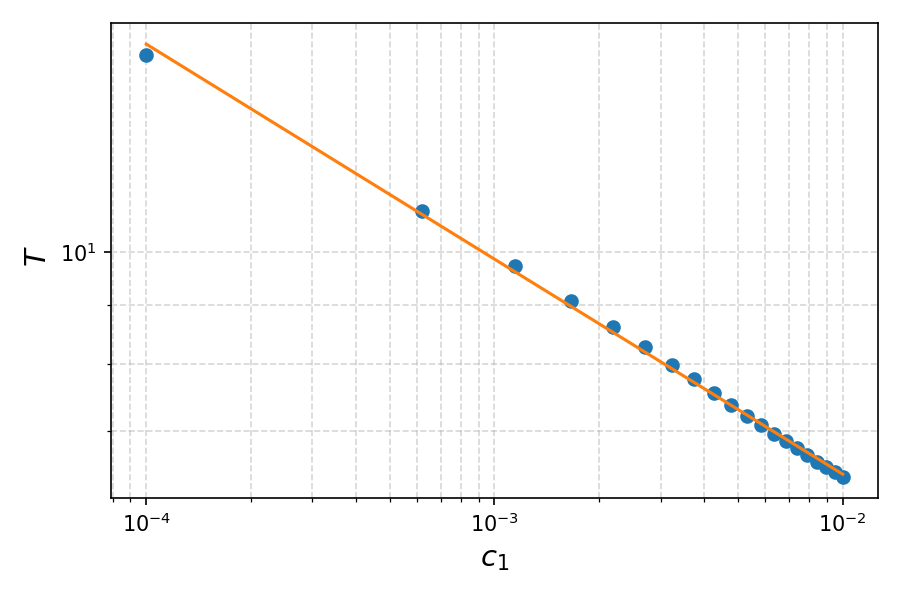}
  \end{subfigure}
  \caption{Left: a trajectory approaching the limit cycle in 3D for the
    parameter values of Table~\ref{table 1}.  Right: the period length $T$
    of the limit cycle in dependence of the cyclin synthesis $c_1$. }
  \label{cpee}
\end{figure}

\subsection{Bifurcation Analysis for the Parameter $c_7$}

From a biological perspective, cyclin begins to degrade when APC ($z_3$) is
activated, so the activation coefficient $c_7$ for APC is another parameter
of interest for us.  We continue to use the parameter values from
Table~\ref{table 1}, but now consider changes of the value of $c_{7}$.
Figure~\ref{in} shows phase portraits and the CDK evolution for the three
different values of $c_{7}$ for which we computed above the eigenvalues.
As already discussed, the $z_{2}$-pseudo-nullcline is independent of
$c_{7}$ and hence is the same in all three phase portraits.  The shape of
the $z_{3}$-pseudo-nullcline does not strongly change with $c_{7}$; it
remains a monotonically increasing curve, but the increase gets smaller for
larger values of $c_{7}$.

\begin{figure}[t]
  \centering
  \begin{subfigure}{0.4\textwidth}
    \includegraphics[width=\textwidth]{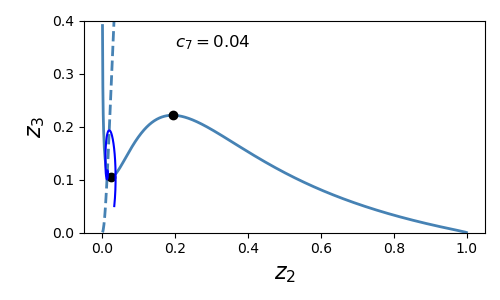}
  \end{subfigure}\qquad
  \begin{subfigure}{0.4\textwidth}
    \includegraphics[width=\textwidth]{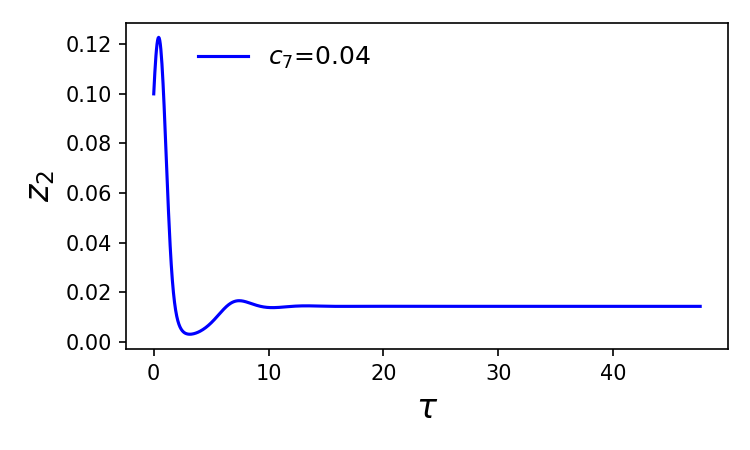}
  \end{subfigure}\\
  \begin{subfigure}{0.4\textwidth}
    \includegraphics[width=\textwidth]{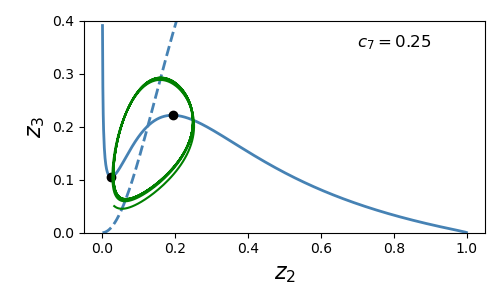}
  \end{subfigure}\qquad
  \begin{subfigure}{0.4\textwidth}
    \includegraphics[width=\textwidth]{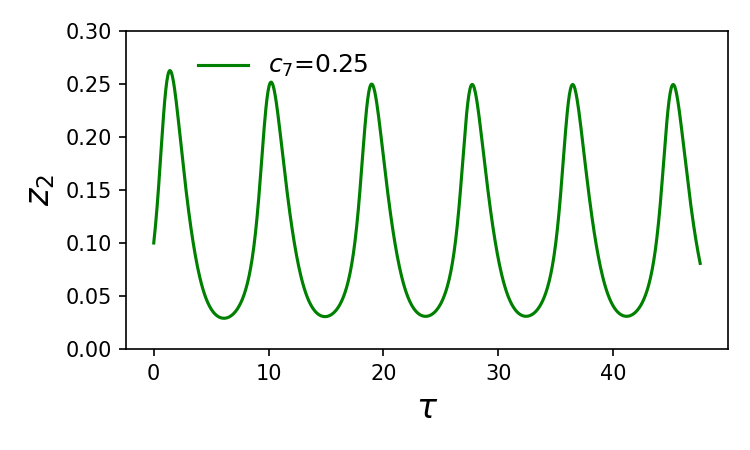}
  \end{subfigure}\\
  \begin{subfigure}{0.4\textwidth}
    \includegraphics[width=\textwidth]{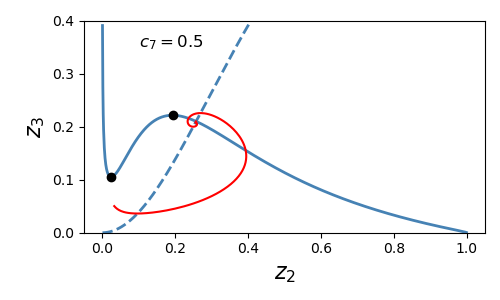}
  \end{subfigure}\qquad
  \begin{subfigure}{0.4\textwidth}
    \includegraphics[width=\textwidth]{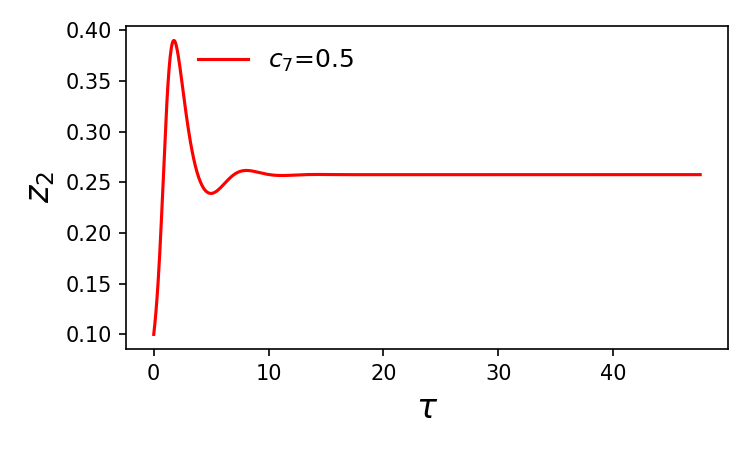}
  \end{subfigure}
  \caption{Phase portraits projected to $z_{2}$-$z_{3}$ plane and CDK
    evolution for three different values of $c_{7}$: top row $c_{7}=0.04$,
    middle row $c_{7}=0.25$, bottom row $c_{1}=0.5$.  The left column shows
    a typical trajectory together with the pseudo-nullclines (solid --
    $z_{2}$, dashed -- $z_{3}$).  In the right column a time plot of
    $z_{2}$ is depicted showing that in the middle row a sustained and in
    the two other rows a strongly damped oscillation occurs.}
  \label{in}
\end{figure}

The middle row corresponds to the values of Table~\ref{table 1} and shows
again the limit cycle we discussed already in the previous section.  But
one can see that for both smaller and larger values of $c_{7}$, the limit
cycle disappears and is replaced by damped oscillations.  Together with the
above presented eigenvalues, this observation indicates that there is a
``Hopf bubble'', i.\,e.\ two subsequent Hopf bifurcations -- one creating a
limit cycle and one destroying it.  The bifurcation diagram in
Figure~\ref{fig:bifc7} confirms this expectation.

Again one can determine several critical values of $c_{7}$.  At
$c_{7}^{(m)}$ the steady state lies at the minimum of the
$z_{2}$-pseudo-nullcline and at $c_{7}^{(M)}$ at the maximum, so that only
for values $c_{7}^{(m)}<c_{7}<c_{7}^{(M)}$ the steady state lies in the
ascending part of the $z_{2}$-pseudo-nullcline.  These two values are
easily computed.  As the $z_{2}$-pseudo-nullcline does not depend on
$c_{7}$, it is given as the graph of $z_{3}=g(z_{2})$ with the same
notation as in the previous section and we can directly determine the
locations $z_{2}^{(m/M)}$ of the two extrema as the zeros of $g'$.  Then we
solve the equations $g(z_{2}^{(m/M)})-h(z_{2}^{(m/M)};c_{7})=0$ to obtain
$c_{7}^{(m)}=0.072379$ and $c_{7}^{(M)}=0.36160$.  The two Hopf points
$c_{7}^{(h)}$, $c_{7}^{(H)}$ are again computed by monitoring the real
parts of the complex eigenvalues and are found to be $c_{7}^{(h)}= 0.12997$
and $c_{7}^{(H)}=0.28328$.  It follows from these critical values that in
the oscillatory regime, the unstable steady state is always in the
ascending part of the $z_{2}$-pseudo-nullcline.  However, as we have
$c_{7}^{(m)}<c_{7}^{(h)}<c_{7}^{(H)}<c_{7}^{(M)}$, we can also find stable
steady states on this ascending part.

\begin{figure}[t]
  \centering
  \includegraphics[width=0.5\textwidth]{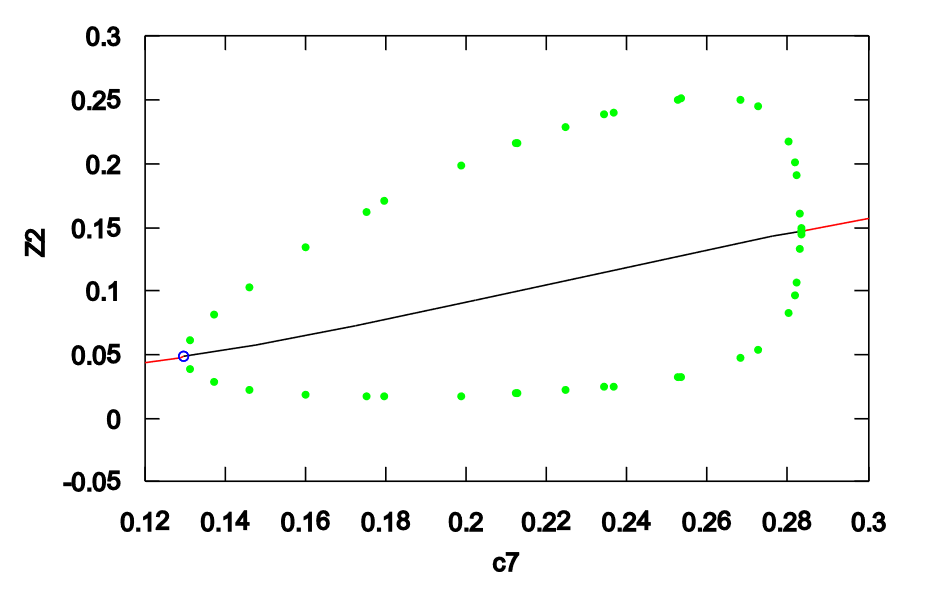}
  \caption{Bifurcation diagram for $c_7$.  Green dots indicate the size of
    the limit cycle; the black line shows unstable steady states, the red
    line stable ones.}
   \label{fig:bifc7}
\end{figure}

\section{Comparison with Experimental Results}\label{exr}

As frequently discussed in the literature -- see e.\,g.\
\citep{deneke2016waves}, \citep{farrell2014egg}, \citep{shindo2021excess}
or \citep{shindo2021modeling} -- the period length in the early
embryogenesis of \emph{D. melanogaster} gets longer and longer
after cycle $10$, and this can be seen by decreasing the synthesis of
cyclin slightly in each cycle (see \citep{shindo2021excess}).  The early
embryonic cell cycles of \emph{D. melanogaster} only consist of S
and M phases, which causes very quick divisions.  However, at some point,
the gap phases G2 are introduced to the cell cycle, which lengthens the
duration of the cycle. For instance, the appearance of G2 in cell cycle 14
is accompanied by the inhibitory phosphorylation of almost all CDK.  The
downregulation of CDK, which results from the destruction of cyclin, causes
this G2 phase, and such extensive inhibitory phosphorylation does not
happen in the early cycles.  The other reason is that after the first 10
cycles, which take place in the inner part of the embryo, the nuclei
migrates into the periphery of the embryo; this might lead to the
lengthening of the cell cycle after cycle $10$.

In our mathematical model, we have so far always treated the cyclin
synthesis, i.\,e.\ the parameter $c_{1}$, as a constant and then obtained
for suitable parameters a sustained oscillation implying in particular that
the period length remains constant in contrast to the experimental
findings.  We now treat $c_{1}$ as a function of time; one may consider
this as a kind of heuristic reduction of a more complicated model involving
also a differential equation for $c_{1}$.  We assume that cyclin is created
in some process and then degraded in another process and that both process
operate on roughly the same time scale.  It is well-known -- e.\,g.\  from
the kinetics of the ABC reaction -- that in such a situation the
concentration of the  intermediate product, in our case cyclin, can be
approximately described by the dynamics
\begin{equation}
  c_1(\tau) = c_{1}^{(0)} \tau e^{-\gamma \tau}\,.
  \label{b1ed}
\end{equation}
Here the factor $\tau$ describes that cyclin has first to be produced and
then degrades exponentially according to the second factor with a constant
rate $\gamma$.   This functions increases from zero until it reaches at
$\tau=1/\gamma$ a maximum value of $e^{-1}c_{1}^{(0)} /\gamma$; then it
declines monotonically towards zero.  By our results above, this means that
we can expect after some transient phase a lengthening of the period of the
oscillation. 

\begin{figure}[t]
  \centering
  \begin{subfigure}{0.45\textwidth}
    \centering
    \includegraphics[width=\textwidth]{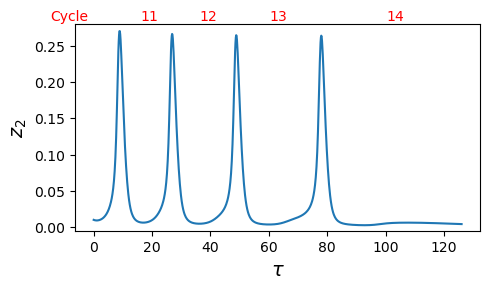}
    \caption{Numerical simulation starting after cycle 10 with
      $c_{1}^{(0)}=0.0001$ and $\gamma=0.055$.}\label{bt}
  \end{subfigure}\qquad
  \begin{subfigure}{0.45\textwidth}
    \centering
    \includegraphics[width=\textwidth]{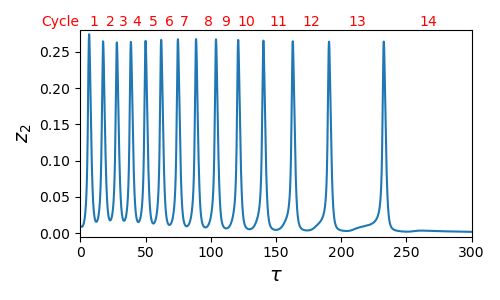}
    \caption{Numerical simulation of first 14 cycles with $c_{1}^{(0)}=0.0001$
      and $\gamma=0.048$.}\label{btt}
  \end{subfigure}\\[0.5\baselineskip]
  \begin{subfigure}{0.45\textwidth}
    \centering
    \includegraphics[width=\textwidth]{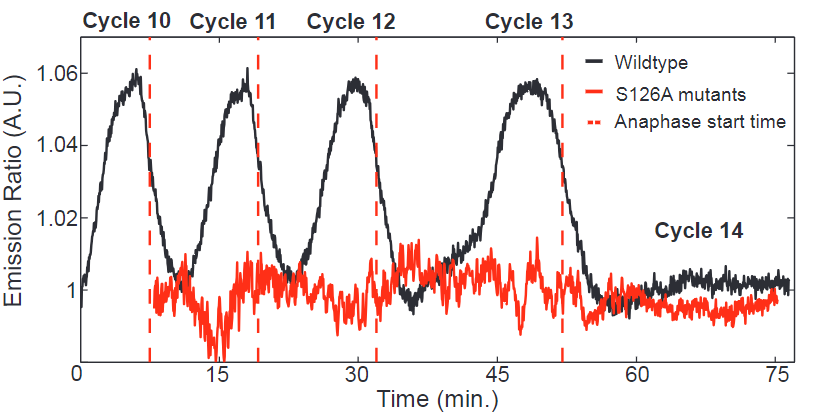}
    \caption{Experimental results from \citet{deneke2016waves}.}\label{all2}
  \end{subfigure}\qquad
  \begin{subfigure}{0.45\textwidth}
    \centering
    \includegraphics[width=\textwidth]{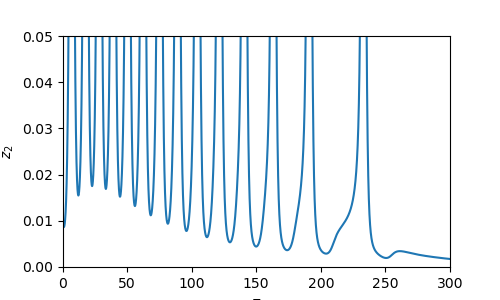}
    \caption{Enlarged bottom part of above plot.}\label{bttz}
  \end{subfigure}
  \caption{Top row: Time evolution of CDK concentration (i.\,e.\ $z_2$) for
    an exponentially decreasing cyclin synthesis~$c_{1}$. Left: cycles 11
    to 14; right: first 14 cycles.  Bottom row: Left: experimental results for cell
    cycles 10 to 14 taken from \citet{deneke2016waves}. Right: Zoom into
    bottom part of \eqref{btt}.}
  \label{ts}
\end{figure}

Figures~\ref{bt} and~\ref{btt} show the results of numerical simulations
with such a variable cyclin synthesis.  In a first simulation, we started
from cycle $1$ using the parameter values $c_{1}^{(0)}=0.0001$ and
$\gamma=0.048$ (Figure~\ref{btt}).  Here one can see that in the first 10
cycles, the period length grows only very moderately, whereas in the last
four cycles a significant increase occurs.  Figure~\ref{bttz} zooms into
the bottom part of the plot and shows that from cycle 2 on, the troughs are
getting deeper and deeper.  This represents exactly the behavior observed
in experiments \citep{deneke2016waves}.  In \citep{edgar1994distinct}, a
similar behavior was reported for CycB, but as CycB is activating CDK, this
is equivalent to our finding.  If one plots the projection of the
trajectory to the $z_2$-$z_3$-plane, it seems, as if the 14 cycles are
almost identical.  But in a 3D plot (not shown here) one can see that over
the 14 cycles the Wee1 level has significantly increased. We are not aware
of any experimental evidence for such a behavior so that it is unclear
whether this represents an artifact of our model or the chosen parameter
values.

Since many studies concentrate on cycles $10$ to $14$, we performed a
second simulation with the parameter values $c_{1}^{(0)}=0.0001$ and
$\gamma=0.048$ to capture the behaviors of these cycles (Figure~\ref{bt}).
For simplicity, we used here actually a simple exponential decline of
$c_{1}$, i.\,e.\ we omitted the factor $\tau$, as it is mainly relevant at
the beginning of embryogenesis.  One can clearly observe a significant
increase in the period length from cycle to cycle.  In both simulations, no
oscillations occur after cycle 14.

Our results align qualitatively quite well with the experimental results of
\citep{deneke2016waves} shown in Figure~\ref{all2}.  In this article, they
demonstrated how the CDK signal shows clear oscillations corresponding to
cell-cycle progression.  While the shape of the oscillation is different in
our simulations compared to the experiment, the periodic behavior is well
reproduced and also the fact that the oscillations stop after cycle~14.
Furthermore, one can in both simulation and experiment clearly distinguish
the S and the M phases: the S phases are characterised by low CDK levels
and the M phases by high CDK levels.  In our model, the S phases are much
longer than the M phases which is in contrast to the experimental findings.
Perhaps this could be corrected by changing some of the other parameters.

\section{Conclusion}\label{concl}

We presented a mathematical model for the cell cycle dynamics in
\emph{D. melanogaster} during early embryogenesis by reviewing
experimental findings and translating them into a biochemically sound model
of CDK regulation with positive and negative feedback loops.  The raw model
consists of six ordinary differential equations depending on fifteen
parameters and three Hill coefficients.  Using exact reduction methods like
conservation laws and dimensional analysis and approximate methods like a
quasi-steady state assumption, we reduced to a model consisting of three
differential equations depending on nine parameters and three Hill
coefficients.  We then performed a bifurcation analysis with respect to two
key parameters, the cyclin synthesis and the activation coefficient for
APC, for proving the existence of regions in the parameter space where the
solutions exhibit limit cycle oscillations as observed in experiments.  We
found a single Hopf bifurcation when changing the cyclin synthesis
entailing that oscillations occur only when its concentration drops below a
certain threshold.  When changing the activation coefficient for APC, a
Hopf bubble appears, indicating that it must take its values in a certain
finite interval for oscillations.

In numerical simulations, we investigated in particular the dependence of
the period length on the cyclin synthesis.  The results support our
biological hypothesis that the period length increases with decreasing
cyclin synthesis by showing a simple inverse power law dependency.  For a
comparison with experimental results, we prescribed a concrete dynamics for
the cyclin synthesis by considering it as a time dependent function.
Qualitatively, the simulation results agreed fairly well with experimental
data reported in the literature.  We could observe for suitable parameter
values the same increase of the period length over the first 14 cycles and
the oscillations stopped after cycle 14.  We could also observe that the
CDK minima got lower in each cycle.  Not so well represented was the shape
of the oscillations.  In particular, the M phases were much shorter than
the S phases.  One can hope that this could be changed by modifications of
other model parameters.

According to \citep{novak1993numerical}, \citep{pomerening2003building} and
references therein, the embryogenesis of \emph{Xenopus laevis} exhibits a
similar type of oscillations.  In yeast, Swe1 regulates Cdc28 in a manner
analogous to the Wee1-mediated Cdk inhibition in \emph{Drosophila
  melanogaster} \citep{causton2015metabolic}.  Our numerical simulations for
Wee1 qualitatively agrees with the experimental results reported for Swe1
of the yeast cell cycle in \citep{causton2015metabolic}.  With small
modifications, our model should thus also work well for the early cell cycles of
certain other organisms.

\bibliography{Dro}
\bibliographystyle{elsarticle-harv}
\end{document}